\documentclass[runningheads]{llncs}
\usepackage[T1]{fontenc}
\usepackage{graphicx}
\newif\iffull
\fulltrue

\usepackage{microtype}
\usepackage{tikzscale}
\usepackage{xcolor}
\usepackage{comment}
\usepackage{caption}
\usepackage{subcaption}
\usepackage{hyperref}
\usepackage{framed}
\usepackage{quoting}
\usepackage{soul}
\usepackage{pgfplotstable}
\usepackage{pgfplots}
\usepackage{threeparttable}
\usepackage{ragged2e}
\usetikzlibrary{patterns}
\usepackage{xspace}
\usepackage{enumitem}
\usepackage{amsmath}
\usepackage{amssymb}
\usepackage{cleveref}
\usepackage{proof}

\usepackage{listings}
\usepackage{algorithm}
\usepackage{algpseudocode}
\usepackage{stmaryrd}
\usepackage[framemethod=tikz]{mdframed}
\definecolor{backcolour}{rgb}{0.95,0.95,0.94}
\definecolor{lightblue}{rgb}{.97, .97, 1}
\usepackage{todonotes}

\lstdefinelanguage{Isabelle}{
    keywords=[1]{definition, record, lemma, where, datatype, type_synonym, function, for, in, if, then, else, return, while, and},
    keywords=[2]{'v},
    keywords=[3]{acq_wr_lock, cl_commit, cl_write_commit, cl_read_invoke},
    alsoletter=':,
    sensitive=true, 
    morecomment=[l]{---} 
}

\lstdefinestyle{Isabellestyle}{
    backgroundcolor=\color{backcolour}, 
    commentstyle=\color{olive}\rmfamily\itshape,
    keywordstyle=[1]\color{blue}\bfseries,
    keywordstyle=[2]\color{violet}\bfseries,
    keywordstyle=[3]\bfseries,
    numberstyle=\tiny\color{gray},
    basicstyle=\linespread{0.85}\ttfamily\small,
    basewidth=0.5em,
    breakatwhitespace=false,         
    breaklines=true,                 
    captionpos=b,
    columns=fixed,
    fontadjust=true,
    frame=single,
    keepspaces=true,                 
    mathescape,
    numbers=left,
    numbersep=5pt,
    rulecolor=\color{backcolour},                  
    showspaces=false,                
    showstringspaces=false,
    showtabs=false,                  
    tabsize=2,
}

\newcommand{\inlsec}[1]{\smallskip\noindent\textbf{#1}}

\newcommand{\invcmts}{inverted commits\xspace}

\newcommand{\ep}{Eiger-PORT\xspace}
\newcommand{\ouralg}{Eiger-PORT+\xspace}

\newcommand{\nats}{\ensuremath{\mathtt{nat}}}

\newcommand{\fun}{\rightarrow}
\newcommand{\map}{\rightharpoonup}

\newcommand{\Skip}{\mathsf{skip}}
\newcommand{\reach}{\mathsf{reach}}
\newcommand{\refines}{\ensuremath{\mathrel{\preccurlyeq}}}

\newcommand{\refmap}[2]{#1_{#2}}

\newcommand{\refmapkvs}[1]{\inlisa{kvs\_of}}
\newcommand{\refmapviews}[1]{\inlisa{views\_of}}

\DeclareTextFontCommand{\isakw}{\rmfamily\bfseries}  
\newcommand{\isaco}[1]{\ensuremath{\mathsf{#1}}}                  
\DeclareTextFontCommand{\isaid}{\rmfamily\itshape}   

\newcommand{\isa}{\isaid}

\newcommand{\keytype}{\inlisa{key}}
\newcommand{\valuetype}{\inlisa{value}}
\newcommand{\verstype}{\inlisa{version}}

\newcommand{\listtype}[1]{\inlisa{#1$\;$list}}
\newcommand{\settype}[1]{\inlisa{#1$\;$set}}
\newcommand{\rwtype}{\{\inlisa{R,W}\}}

\newcommand{\TxID}{\isa{TxID}}

\newcommand{\SO}{\isaco{SO}}

\newcommand{\WR}{\isaco{WR}}

\newcommand{\kvs}{\mathcal{K}}
\newcommand{\views}{\mathcal{U}}
\newcommand{\fprint}{\mathcal{F}}

\newcommand{\readonly}{\isaco{rdonly}}
\newcommand{\visTx}{\isaco{visTx}}

\newcommand{\LWW}{\isaco{LWW}}
\newcommand{\wellformed}{\isaco{wf}}

\newcommand{\canCommit}{\isaco{canCommit}}

\newcommand{\RYW}{\isaco{RYW}}

\newcommand{\isomodel}[1]{\mathcal{I}_{#1}}
\newcommand{\restricted}[1]{\ensuremath{\widehat{#1}}}

\newcommand{\epp}{\ensuremath{\mathtt{EPP}}}
\newcommand{\eppordered}{\restricted{\epp{}}}

\newcommand{\TCCv}{\isaco{TCCv}}

\DeclareTextFontCommand{\inlisa}{\ttfamily}   

\newcommand{\hlc}[2][yellow]{{%
    \colorlet{foo}{#1}%
    \sethlcolor{foo}\hl{#2}}%
}

\definecolor{shadecolor}{gray}{0.95}
\newenvironment{shadedquotation}
 {
  \begin{shaded*}
  \vspace{-2pt}
  \quoting[indentfirst=false, leftmargin=1pt, vskip=0pt]
 }
 {\endquoting
 \vspace{-2pt}
 \end{shaded*}
}

\begin{document}
\title{Pushing the Limit: Verified Performance-Optimal Causally-Consistent Database Transactions}
\titlerunning{Verified Performance-Optimal Causally-Consistent Database Transactions}
%
\author{Shabnam Ghasemirad
\and Christoph Sprenger \and
Si Liu\and \\  Luca Multazzu \and David Basin}
\authorrunning{S. Ghasemirad et al.}
%
\institute{ETH Zurich, Switzerland
}
\maketitle              
\begin{abstract}
Modern web services crucially rely on high-performance distributed databases, where concurrent transactions are isolated from each other using concurrency control protocols. Relaxed isolation levels, which permit more complex concurrent behaviors than strong levels like serializability, are used in practice for higher performance and availability. 

In this paper, we present \ouralg, a concurrency control protocol that achieves a strong form of causal consistency, called TCCv (Transactional Causal Consistency with convergence). We show that \ouralg also provides performance-optimal read transactions in the presence of transactional writes, thus refuting an open conjecture that this is impossible for TCCv. We also deductively verify that \ouralg satisfies this isolation level by refining an abstract model of transactions. This yields the first deductive verification of a complex concurrency control protocol. Furthermore, we conduct a performance evaluation showing \ouralg's superior performance over the state-of-the-art.


\end{abstract}
\section{Introduction}

Modern web services are built on top of high-performance database systems operating in partitioned, geo-distributed environments. These systems provide distributed transactions that group the users' read and write requests.
To balance data consistency and system performance, 
databases provide a spectrum of \emph{isolation levels} (I in ACID: Atomicity, Consistency, Isolation, and Durability~\cite{DBLP:books/mg/SKS20}), 
defining the degree of separation between concurrent transactions.
Isolation is enforced by \emph{concurrency control protocols} (also called \emph{transaction protocols}).

Many applications, such as social networks, opt for weak isolation levels to avoid the performance overhead of stronger levels like serializability~\cite{serializability}.  
These weaker guarantees allow distributed transactions to remain functional even during network partitions, while still providing useful properties. 
Notably, \emph{transactional causal consistency} (TCC) represents a successful integration of ideas from the distributed computing and database communities. It extends causal consistency~\cite{CausalMemory:DC1995,CC:PPoPP2016}---the strongest consistency level achievable in an always-available system~\cite{Limitations:PODC2015}---by incorporating transactional guarantees. 
The past decade has seen sustained efforts in designing databases supporting performant causally-consistent distributed transactions~\cite{Eiger:NSDI2013,Cure:ICDCS2016,FriendFoe:VLDB2018,Slowdown:NSDI2017,NOC:OSDI2020,OCC:TPDS2021},
along with their growing adoption in industry~\cite{Neo4j,ElectricSQL,CosmosDB}.
Nearly all of these systems provide a stronger variant of TCC,
known as TCCv~\cite{Cure:ICDCS2016,Eiger:NSDI2013}, that includes \emph{data convergence} requiring views across different clients to eventually converge to the same state.

\looseness=-1
In this paper, we present a case study on developing and verifying a performance-optimal, causally-consistent, database transaction protocol. 
Our protocol, \ouralg, provides TCCv, for which we give a formal proof. 
Our work faced two challenges.
First, it is not a priori clear that such an isolation guarantee is achievable for a performance-optimal protocol. In fact, Lu et al.~\cite{NOC:OSDI2020} conjectured that for distributed performance-optimal read-only transactions (PORTs) in the presence of transactional writes, TCC (without convergence) is the strongest achievable isolation level. They presented the \ep protocol, which provides this guarantee. In this paper, we constructively refute their conjecture by designing our novel protocol, \ouralg, which achieves the stronger TCCv guarantee.

\looseness=-1
Second, transaction protocols are notoriously hard to get right, as witnessed by numerous design-level isolation errors in production databases~\cite{elle,polysi,plume,jepsen-analyses,txcheck}, and even in protocols that have undergone pen-and-pencil proofs~\cite{DBLP:conf/wadt/Olveczky16} and model-checking analysis~\cite{osdi23}. 
We thus aim for a full deductive verification of \ouralg,  covering \emph{all} possible behaviors. To our knowledge, the deductive verification of transaction protocols for weak isolation levels, which exhibit complex concurrent behaviors, has not been attempted so far. Previous efforts in this area have focused on simple textbook protocols like two-phase locking achieving serializability or employed model checking, which requires bounding the number of processes and transactions. We address this challenge using our Isabelle/HOL framework~\cite{GhasemiradLiuSprenger+-VLDB25} built around Xiong et al.'s abstract transaction model~\cite{DBLP:conf/ecoop/XiongCRG19}. We formalize \ouralg and show that it satisfies TCCv using reduction~\cite{DBLP:journals/cacm/Lipton75} in combination with a refinement of the TCCv instance of the abstract transaction model.

Furthermore, we implement and deploy \ouralg, along with \ep and its precursor Eiger,  and conduct a comprehensive performance comparison of these three protocols.
Our evaluation demonstrates  \ouralg's superior performance in terms of system throughput and latency across various scenarios, e.g., with a growing number of clients and servers.  

The complete formal development accompanying this paper, including all definitions and proofs, as well as a protocol implementation are available at~\cite{ghasemirad_2025_14622074}.

\medskip
\inlsec{Contributions} Overall, we see our contributions as three-fold:
\begin{itemize} 
\item \textbf{Conjecture refutal.}
We formally refute the conjecture that TCC is the strongest achievable isolation level for PORTs in the presence of transactional writes by designing a protocol, \ouralg, that provably achieves TCCv.

\item \textbf{Proof of correctness.}
We model \ouralg in Isabelle/HOL and verify its correctness by showing that its behavior conforms to the TCCv instance of the abstract transaction model~\cite{DBLP:conf/ecoop/XiongCRG19}. 
This represents the first complete formal verification of a complex distributed database transaction protocol.

\item \textbf{Superior performance.}
We deploy \ouralg, along with two state-of-the-art causally-consistent transaction protocols, in a cluster and evaluate their performance.
Our experimental results demonstrate \ouralg's superior performance, with both lower latency and higher throughput. 

\end{itemize}

\section{Background}

\subsection{Distributed Database Transactions}
\label{subsec:distributed-database-transactions}

In a distributed database, vast amounts of data are
split up and 
stored across multiple servers, also called partitions. 
User requests are submitted as database transactions, initiated by front-end clients.
Each client executes the transactions in its own session, where 
a transaction comprises a sequence of read and/or write operations on data items (or keys) distributed across partitions.

\smallskip
\inlsec{Isolation levels}
Distributed databases offer various isolation levels, 
differing on how they
balance data consistency and system performance. 
Figure~\ref{fig:isolation-hierarchy} 
shows a spectrum of practically relevant 
isolation levels, ranging from weaker ones like Read Committed, through various forms of transactional causality, to stronger guarantees such as Serializability. We briefly explain Read Atomicity and two variants of transactional causality, which are the focus of this work.

\begin{figure}[t]
\begin{center}
   \includegraphics[width=.9\columnwidth]{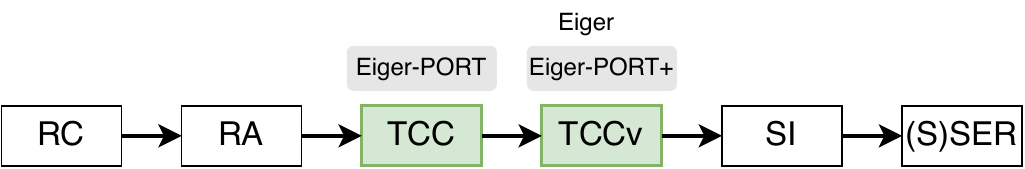}
\end{center}
\captionsetup{skip=0pt}
   \caption{A spectrum of isolation levels.  $A \to B$ means  $A$ is  weaker than $B$.
    RC: read committed~\cite{si};
	RA: read atomicity~\cite{ramp};
	TCC: transactional causal consistency~\cite{NOC:OSDI2020}, provided by \ep~\cite{NOC:OSDI2020};
    TCCv: TCC with convergence~\cite{Cure:ICDCS2016,Eiger:NSDI2013},  offered by Eiger and our \ouralg; 
	SI: snapshot isolation~\cite{si}; 
	(S)SER: (strict) serializability~\cite{serializability}. 
Protocols supporting PORTs are highlighted in 
\hlc[gray!20]{gray}.
   }
   \label{fig:isolation-hierarchy}
\end{figure}

\begin{description}[leftmargin=15pt]
\item[Read Atomicity (RA).] This is also known as \emph{atomic visibility}, requiring that all or none of a transaction’s updates are observed by other transactions. 
It prohibits \emph{fractured reads} anomalies, such as Carol only observing one direction of a new (bi-directional) friendship between Alice and Bob in a social network.

\item[Transactional Causal Consistency (TCC).] In addition to RA, this level requires that two causally related transactions appear to all client sessions in the same causal order~\cite{CausalMemory:DC1995,CC:PPoPP2016}. 
It prevents \emph{causality violations}, such as Carol observing Bob's response to Alice's message without seeing the message itself.

\item[TCC with Convergence (TCCv).] With TCC, different clients may observe causally unrelated transactions in different orders. 
TCCv's \emph{convergence} property prevents this by requiring these clients' views to converge to the same state~\cite{COPS:SOSP2011,Cure:ICDCS2016}. 
For example, this prevents confusion created by Alice and Bob independently posting ``Let's meet at my place'' in a road trip planner.
In practice, most causally-consistent databases provide convergence.
\end{description}

\inlsec{Performance-optimal read-only transactions} 
\looseness=-1 NOCS~\cite{NOC:OSDI2020} is the state-of-the-art impossibility result that captures conflicts between 
distributed transactions' performance and their isolation guarantees. 
NOCS proves that read-only transactions cannot complete with \textbf{N}on-blocking communication in \textbf{O}ne round and \textbf{C}onstant-size metadata, while achieving \textbf{S}trict serializability (SSER). 
At best, three of the four NOCS properties can be satisfied. In particular, protocols satisfying the NOC properties (under isolation levels weaker than SSER) are said to provide \emph{performance-optimal} read-only transactions (PORTs). The NOCS authors~\cite{NOC:OSDI2020} introduce the \ep protocol providing TCC and PORTs. 
They also state the following conjecture, which has remained unresolved for four years. 

\begin{shadedquotation}\label{conjecture}
\textbf{\emph{Conjecture}.} TCC is the strongest isolation level achievable for PORTs in the presence of transactional writes.
\end{shadedquotation}

Recent studies show that write-heavy workloads involving transactional writes are more prevalent than previously assumed and are expected to become increasingly prominent~\cite{DBLP:journals/tos/YangYR21}. 
This, along with the practical significance of TCCv, motivates our work on
refuting the above conjecture and 
pushing the boundary.

\subsection{Transition Systems and Refinement}
\label{subsec:lts-refinement}

\looseness=-1
We use \emph{labeled transition systems} (LTSs) to model database protocols and the abstract transaction model. 
An LTS $\mathcal{E} = (S, I, \{\xrightarrow{e} \, \mid e \in E\})$ consists of a set of states $S$, a non-empty set of initial states $I \subseteq S$, and transition relations $\xrightarrow{e} \;\subseteq S \times S$, one for each event $e \in E$. We assume an idling event $\Skip \in E$ with $s \xrightarrow{\Skip} s$.
We often define the relations $\xrightarrow{e}$ using \emph{guard} predicates $G_e$ and \emph{update} functions $U_e$ by $s \xrightarrow{e} s'$ if and only if $G_e(s) \land s' = U_e(s)$. 
A state~$s$ is \emph{reachable} if a sequence of transitions leads from an initial state to~$s$. We denote the set of reachable states of $\mathcal{E}$ by $\reach(\mathcal{E})$. A set of states $J$ is an \emph{invariant} if $\reach(\mathcal{E}) \subseteq J$.  

\emph{Refinement} relates two LTSs 
$\mathcal{E}_i = (S_i, I_i, \{\xrightarrow{e}_i \, \mid e\in E_i\})$, for $i \in\{1,2\}$. Given \emph{refinement mappings} $r\!: S_2 \fun S_1$ and $\pi\! : E_2 \fun E_1$ between the LTSs' states and events, we say $\mathcal{E}_2$ refines $\mathcal{E}_1$, written $\mathcal{E}_2 \refines_{r,\pi} \mathcal{E}_1$, if
(i) $r(s) \in I_1$ for all $s \in I_2$ and 
(ii) $r(s) \xrightarrow{\pi(e)}_1 r(s')$ whenever $s \xrightarrow{e}_2 s'$. 
Using guards and updates, (ii) reduces to two proof obligations: assuming $G^2_e(s)$ prove (a) $G^1_{\pi(e)}(r(s))$ (\emph{guard strengthening}) and (b) $r(U^2_e(s)) = U^1_{\pi(e)}(r(s))$ (\emph{update correspondence}).
Refinement guarantees the inclusion of sets of reachable states (modulo~$r$), i.e., $r(\reach(\mathcal{E}_2)) \subseteq \reach(\mathcal{E}_1)$, where $r$ is applied to each element of $\reach(\mathcal{E}_2)$.
Refinement proofs often require invariants 
to strengthen the refinement mapping.

\subsection{An Abstract Transaction Model}
\label{subsec:abstract-transaction-model}

Xiong et al.~\cite{DBLP:conf/ecoop/XiongCRG19} introduced a centralized operational model for atomic transactions operating on distributed multi-versioned key-value stores (KVSs), where the database stores data as key-value pairs and each key may be mapped to multiple versions for increased data availability. This model can be instantiated to different isolation guarantees, including RA, TCCv, and SSER (cf.~\Cref{fig:isolation-hierarchy}). They prove the equivalence of these model instances to their declarative counterparts based on abstract executions. The model can be used to prove the correctness of both concurrency control protocols and client programs. We have formalized this framework in Isabelle/HOL and extended it with an extensive library of lemmas supporting protocol correctness proofs~\cite{GhasemiradLiuSprenger+-VLDB25}.

\looseness=-1
Xiong et al.'s model is formulated as an LTS, called the \emph{abstract transaction model}, which abstracts the protocols' distributed collection of KVSs (each representing a shard and/or replica) into a single (centralized) multi-versioned KVS $\kvs\!: \keytype \fun \listtype{\verstype}$ that maps each key to a list of \emph{versions}. Each version~$\kvs(k,i)$ of a key~$k$ at the list index $i$ records (i) the value $v$ stored, (ii) the writer transaction $t$ that has produced this version, and (iii) the reader set $T$, i.e., the set of transactions that have read this version. The pairs $(t, t')$ for any $t' \in T$ are called \emph{write-read dependencies}. The relation $\WR_{\kvs}$ contains all such pairs.
The fact that, in a real, distributed system, each client $cl$ has a different partial \emph{client view} of $\kvs$ is modeled by explicitly representing these views in the model's configurations as mappings $\views(cl) \!: \keytype \fun \settype{\nats}$. This describes, for each key, the set of versions (denoted by list indices) visible to the client. Clients are assumed to process transactions sequentially. The \emph{session order} relation~$\SO$ captures the order of their transactions. 

The model assumes the \emph{snapshot property}, ensuring that each transaction reads and writes at most one version of each key. Hence, transactions can be represented by a \emph{fingerprint} $\fprint \! : \keytype \times \rwtype \map \valuetype$, which maps each key and operation (read or write) to at most one value. It also assumes that views are \emph{atomic}, i.e., clients observe either all or none of a transaction's effects. 
These properties together establish \emph{atomic visibility}, also called Read Atomicity (RA) (cf.~\Cref{subsec:distributed-database-transactions}), as the model's baseline isolation guarantee.

The model has two events (plus $\Skip$): \emph{commit}, which atomically executes an entire transaction,  
and \emph{view extension}, which monotonically extends a client's view of the KVS. 
The commit event's executability depends on the isolation guarantee.
Here, we focus on the model's TCCv instantiation, called $\isomodel{\TCCv}$.

\smallskip
\inlsec{Commit} 
The commit event's transition relation for TCCv is defined by:
\[
\renewcommand{\arraystretch}{1.2}
\infer{
  (\kvs, \views) \xrightarrow{\isaco{commit}(cl, \isa{sn}, u, \fprint)}_{\TCCv} 
  (\kvs', \views[cl \mapsto u'])
}{
  \begin{array}{c}
  \views(cl) \sqsubseteq u \quad 
  \canCommit_{\TCCv}(\kvs, u, \fprint) \quad
  u \sqsubseteq u' \quad
  \RYW(\kvs, \kvs', u') 
\\
  \LWW(\kvs, u, \fprint) \quad       
  \wellformed(\kvs, u) \quad 
  \wellformed(\kvs', u') 
\\
  t^{cl}_{sn} \in \isaco{nextTxids}(\kvs, cl) \quad
  \kvs' = \isaco{UpdateKV}(\kvs, t^{cl}_{sn}, u, \fprint)
  \end{array}
}
\]

The transition in the conclusion updates the configuration $(\kvs, \views)$ to the new configuration $(\kvs', \views[cl \mapsto u'])$, where $\kvs'$ is the updated KVS and $\views[cl \mapsto u']$ updates the client $cl$'s view  to $u'$. Both $\kvs'$ and $u'$ are determined by the rule's premises, which act as the event's guards with the following meanings:
\begin{itemize}[leftmargin=15pt]
\item $\views(cl) \sqsubseteq u$: This condition allows one to extend the client $\isa{cl}$'s current view to a (point-wise) larger one before committing.

\item $\canCommit_{\TCCv}(\kvs, u, \fprint)$: This is the central commit condition, which ensures that it is safe to commit a transaction at the TCCv isolation level. It requires that the set of visible transactions $\visTx(\kvs,u)$ (i.e., the writers of the versions that the view $u$ points to) is \emph{closed} under the relation $\SO \cup \WR_{\kvs}$, i.e., 
\begin{equation}
\label{eq:closedness} 
  ((\SO \cup \WR_{\kvs})^{-1})^{+}(\visTx(\kvs,u)) \subseteq \visTx(\kvs,u) \cup \readonly(\kvs).\footnote{This condition differs from the one presented in~\cite{DBLP:conf/ecoop/XiongCRG19}, but is equivalent.}  
\end{equation}
In other words, following the \emph{causal dependency} relation $\SO \cup \WR_{\kvs}$ backwards from visible transactions, we only see visible or read-only transactions.

\item $u \sqsubseteq u'$: This condition captures the \emph{monotonic reads} session guarantee, i.e., the view $u'$ extends the view $u$.

\item \looseness=-1 $\RYW(\kvs, \kvs', u')$: This condition expresses the \emph{read-your-writes} (RYW) session guarantee, stating that each client sees all versions previously written by itself.  

\item $\LWW(\kvs, u, \fprint)$: This captures the \emph{last-write-wins} conflict resolution policy, whereby a client reads each key's latest version in its view. 

\item $\wellformed(\kvs, u)$ and $\wellformed(\kvs, u')$: This requires that the views $u$ and $u'$ are \emph{wellformed}, i.e., they are atomic and contain indices that point to existing versions.

\item $t^{cl}_{sn} \in \isaco{nextTxids}(\kvs, cl)$: Transaction identifiers $t^{cl}_{sn} \in \TxID$ are indexed by the issuing client $cl$ and a (monotonically increasing) sequence number $sn$. This condition obtains a fresh transaction ID $t^{cl}_{sn}$, where the sequence number $sn$ is larger than any of the client $cl$'s sequence numbers used in~$\kvs$.

\item $\kvs' = \isaco{UpdateKV}(\kvs, t^{cl}_{sn}, u, \fprint)$: The KVS $\kvs'$ is obtained from $\kvs$ by adding the operations described by the fingerprint $\fprint$: the writes append a new version with the writer ID $t^{cl}_{sn}$ to the respective key's version list, and the reads add $t^{cl}_{sn}$ to the respective versions' reader sets.

\end{itemize}

\smallskip
\inlsec{View extension}
The view extension event for TCCv is defined by the rule:
\[
\renewcommand{\arraystretch}{1.2}
\infer{
  (\kvs, \views) \xrightarrow{\isaco{xview}(cl, u)}_{\TCCv} 
  (\kvs, \views[cl \mapsto u])
}{
  \begin{array}{c c}
  \views(cl) \sqsubseteq u & \quad \wellformed(\kvs, u)
  \end{array}
}
\]
and extends a client $cl$'s view from $\views(cl)$ to a wellformed view $u$. It abstractly models that additional versions of certain keys become visible to the client.

\section{Eiger-PORT+: An Overview}

\looseness=-1
Causally-consistent transactions have attracted the attention of both academia and industry in recent years.
Eiger~\cite{Eiger:NSDI2013} is among the first distributed databases providing TCCv. 
\ep~\cite{NOC:OSDI2020} improves Eiger's overall performance by optimizing its read-only transactions while sacrificing data convergence, thus allowing diverging 
client views of concurrent conflicting writes (to the same keys) as in TCC.
We show that this sacrifice is unnecessary and design \ouralg based on \ep. \ouralg provides TCCv with both read-only and write-only transactions.\footnote{\ouralg's pseudocode, together with its description, is given in 
\iffull \Cref{app:epp}.
\else 
\fi} 
The key idea of achieving convergent client views is to share with clients the \emph{total order} of versions on a server, which has already been established by uniquely assigned timestamps across versions.
In contrast, \ep constructs individual, possibly different, orders per client.
\ouralg's read-only transactions satisfy the NOC properties and are therefore performance-optimal (PORT). This is achieved in the same way as for \ep. In particular, both protocols' read operations use only a fixed number of timestamps as metadata.
We now give a high-level description of both \ep and \ouralg, starting with their commonalities and then highlighting their differences. 

\subsection{Timestamps}
Both distributed transaction protocols leverage \emph{timestamp-based} concurrency control. The timestamps are based on Lamport clocks~\cite{lamportclk}, which clients and servers maintain and update with each local or communication event. Whenever a transaction commits a new version, the current clock reading is paired with the transaction's client ID to generate a globally unique commit timestamp. The lexicographic order of these pairs induces a total order on commit timestamps.

Each server maintains a \emph{local safe time} $\inlisa{lst}$ that corresponds to the minimum of the uncommitted transactions' timestamps or, if there is none, the maximum committed timestamp for that server.
Each client maintains a variable $\inlisa{lst\_map}$, which maps server IDs to their latest known $\inlisa{lst}$ value, and a \emph{global safe time} $\inlisa{gst}$. The latter is updated to the minimum timestamp in $\inlisa{lst\_map}$ when a read-only transaction starts and acts as the \emph{stable frontier} for that client: all transactions with earlier timestamps are guaranteed to be committed on the server.
Each read sent to a server includes $\inlisa{gst}$ as a read timestamp, which is used to safely read a committed version with a timestamp lower than the $\inlisa{gst}$, or the client's latest own write (to achieve RYW), if its timestamp is higher than the $\inlisa{gst}$.

\subsection{Read and Write Transactions}

Write transactions are similar in both protocols and follow a variant of the two-phase commit (2PC) protocol that always commits~\cite{Eiger:NSDI2013}.
In the \emph{prepare} phase, each timestamped write is sent to the corresponding partition, which adds the write to its local data store as a pending version. In the \emph{commit} phase, each partition sets the version as committed, along with its commit timestamp.
However, the two protocols differ in how they handle read transactions in the absence of an own write newer than $\inlisa{gst}$. While \ouralg always reads the \emph{latest} version below $\inlisa{gst}$ in this case, \ep searches for the latest version below $\inlisa{gst}$ that either has no write conflicts or is written by a different client. This backward scan is presumably done to maintain read atomicity (RA). 
However, we show that this scan is unnecessary for RA, can harm performance, and cause client view divergence. Consequently, in addition to providing convergence, \ouralg improves performance by eliminating this scan's overhead. 

\begin{figure}[t]
  \centering
  \begin{subfigure}[t]{.063\linewidth}
    \includegraphics[width=\linewidth]{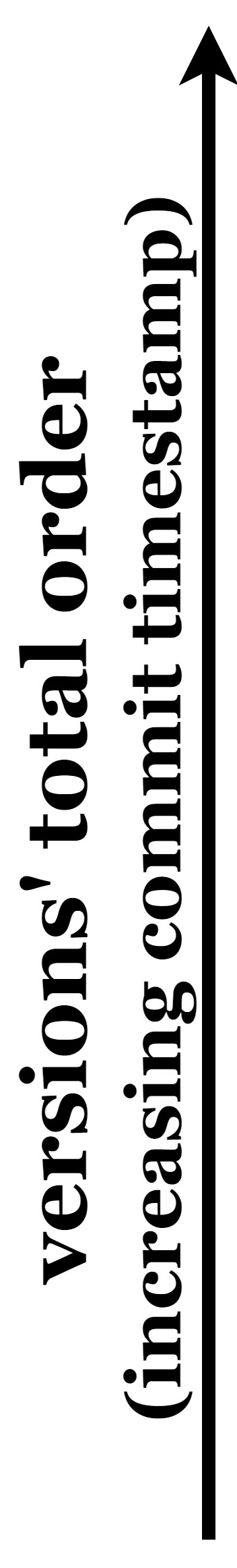}
  \end{subfigure}
  \begin{subfigure}[t]{.3\linewidth}
    \centering\includegraphics[width=\linewidth]{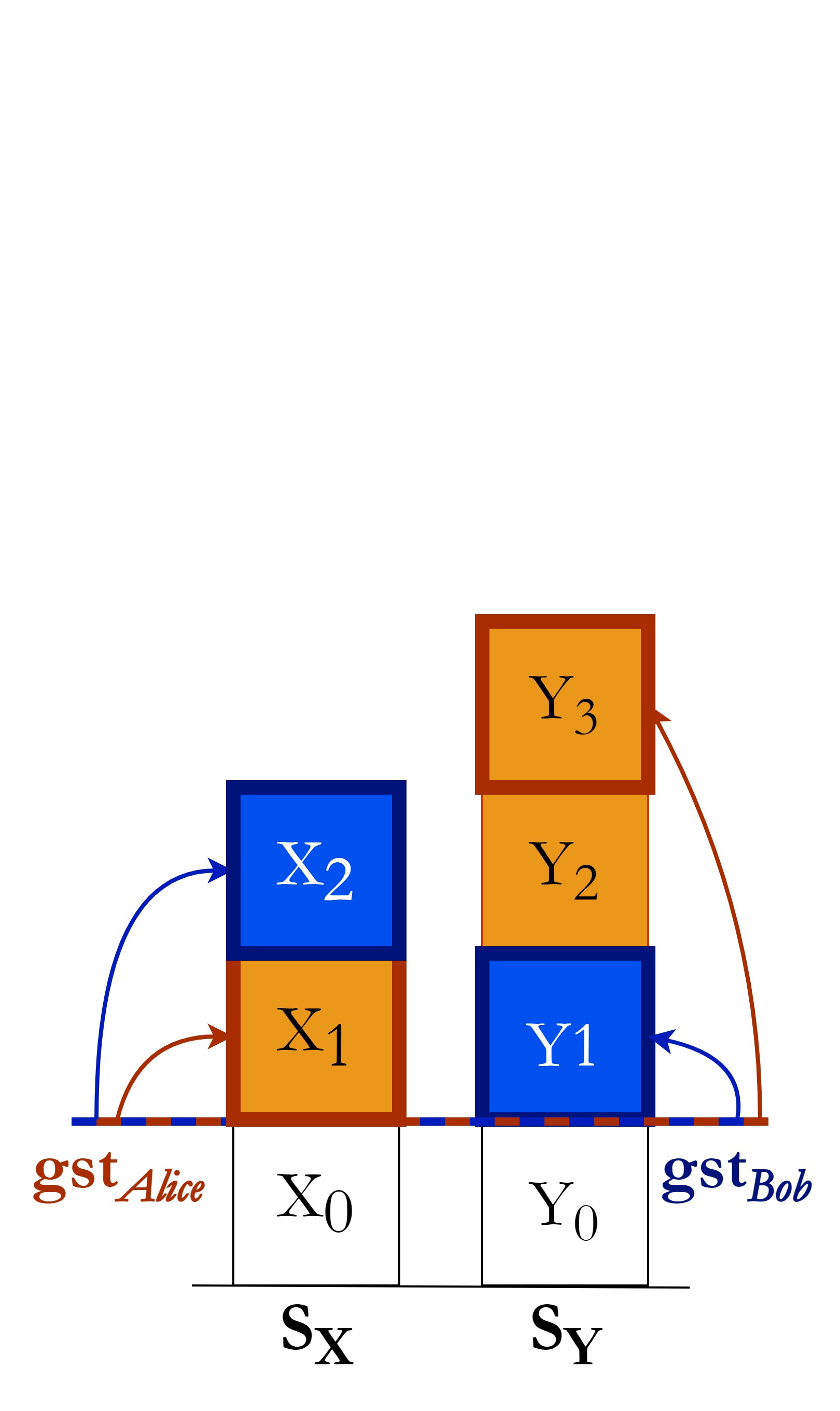}
    \caption{\footnotesize Initial}
  \end{subfigure}
  \begin{subfigure}[t]{.3\linewidth}
    \centering
    \includegraphics[width=\linewidth]{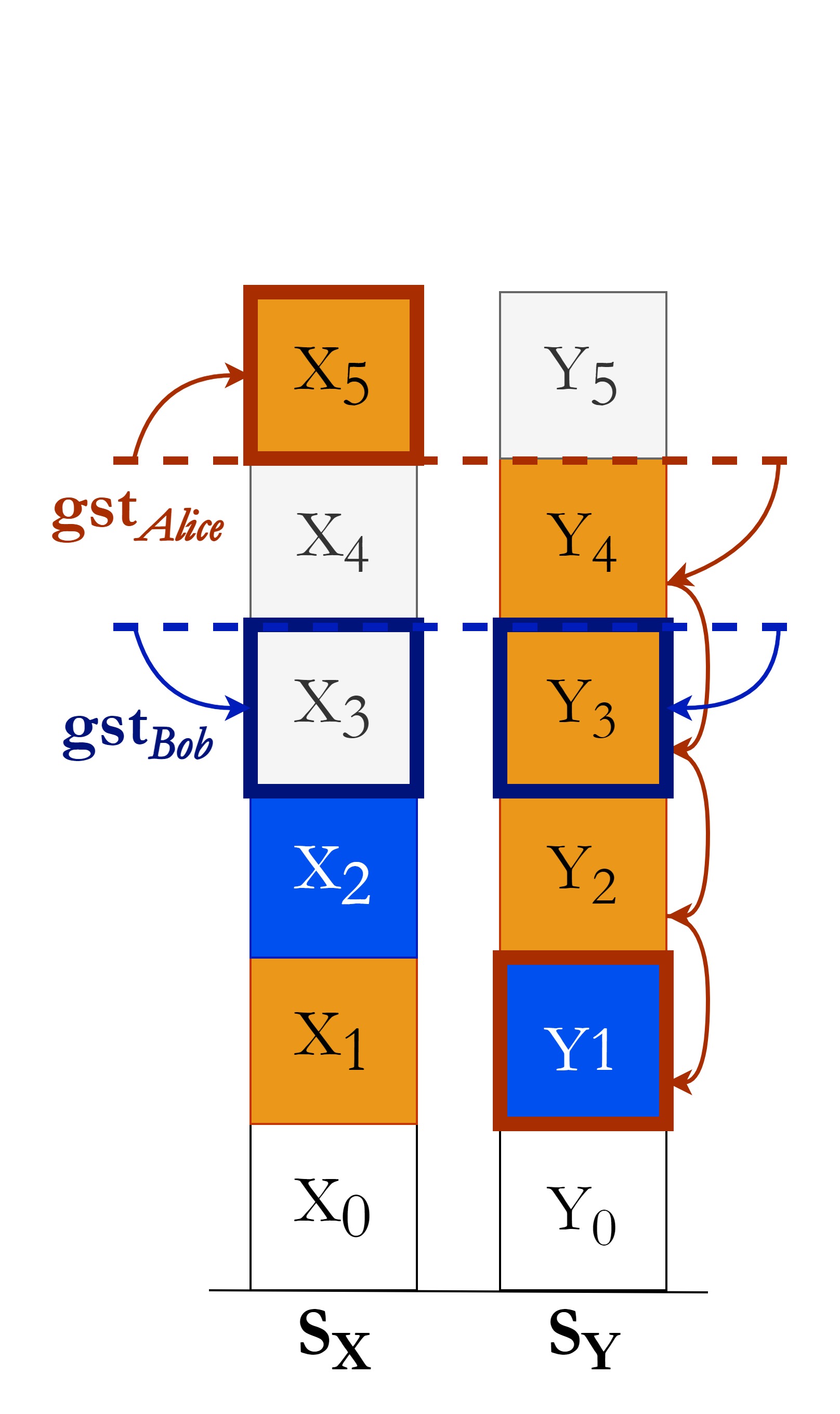}
    \caption{\footnotesize Eiger-PORT}
  \end{subfigure}
  \begin{subfigure}[t]{.3\linewidth}
    \centering
    \includegraphics[width=\linewidth]{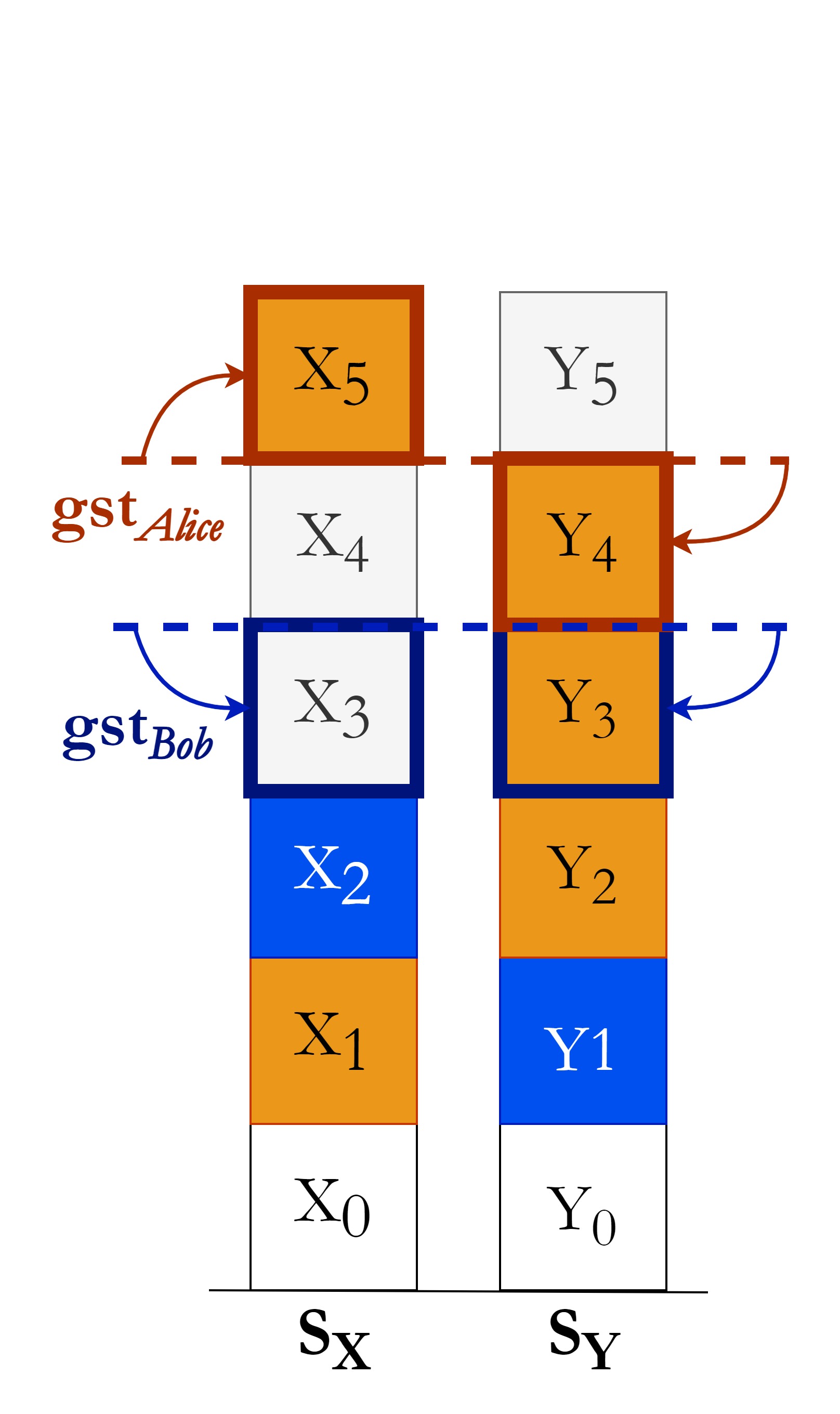}
    \caption{\footnotesize Eiger-PORT+}
  \end{subfigure}
  \captionsetup{skip=5pt}
   \caption{Alice and Bob reading (shown by arrows) from servers storing X and Y, illustrating convergence in \ouralg and lack thereof in \ep. Each square represents a version, and its color determines the version's writer, where orange, blue, and gray correspond to Alice, Bob, and other clients respectively.
   }
   \label{fig:ep_vs_epp}
\end{figure}

\begin{example}
    \Cref{fig:ep_vs_epp} illustrates the difference between \ep and \ouralg in reading versions.
    In~\Cref{fig:ep_vs_epp}a, Alice (orange) and Bob (blue) have written some versions on servers (partitions storing keys) X and Y, and their $\inlisa{gst}$s are~$0$. 
    Alice and Bob each perform a transaction 
    reading from servers X and Y.
    In both protocols, Alice reads $\{X_1, Y_3\}$ and Bob $\{X_2, Y_1\}$ according to RYW.
    
    As new versions are added to the servers, the $\inlisa{gst}$s advance to higher timestamp values. Assume that versions $Y_1$ to $Y_4$ are conflicting writes, i.e., the transactions writing $Y_2$ to $Y_4$ had already started when $Y_1$ was committed. Alice and Bob then again read 
    keys X and Y. 
    Bob behaves the same in both protocols (cf.~Figures \ref{fig:ep_vs_epp}b and \ref{fig:ep_vs_epp}c): He reads the last committed versions below its $\inlisa{gst}$, $\{X_3, Y_3\}$, as he has no writes above its $\inlisa{gst}$ and the read versions are written by other clients.
    However, Alice's behavior differs in the two protocols. Given Alice's new $\inlisa{gst}$, in \ep she reads $\{X_5, Y_1\}$ (cf.~\Cref{fig:ep_vs_epp}b), while in \ouralg she reads $\{X_5, Y_4\}$ (cf.~\Cref{fig:ep_vs_epp}c). As $X_5$ is her newest own write on X, this is the same for both protocols. But Alice has no newer own writes to Y and the latest committed version, $Y_4$, is Alice's write and has write conflicts. Thus,
    \ep performs a scan to find the latest version written by a \emph{different} client below Alice's $\inlisa{gst}$, i.e., $Y_1$, while \ouralg reads $Y_4$, irrespective of its writer or write conflicts.

    Hence, in \ep, Alice reads $Y_3$ before $Y_1$,
    while Bob reads $Y_1$ before~$Y_3$, 
    which results in diverging client views. This behaviour is allowed by TCC where convergence is not required. In contrast, in \ouralg, Bob reads $Y_1$ before $Y_3$ and Alice reads $Y_3$ before $Y_4$, both agreeing with the same convergent order, i.e., the versions' total order established on the servers ($Y_1$ < $Y_3$ < $Y_4$). 
\end{example}

\section{Formal Modeling and Verification}
\label{sec:eigerport-plus}

\begin{figure}[t]
\begin{center}
   \includegraphics[width=.85\columnwidth]{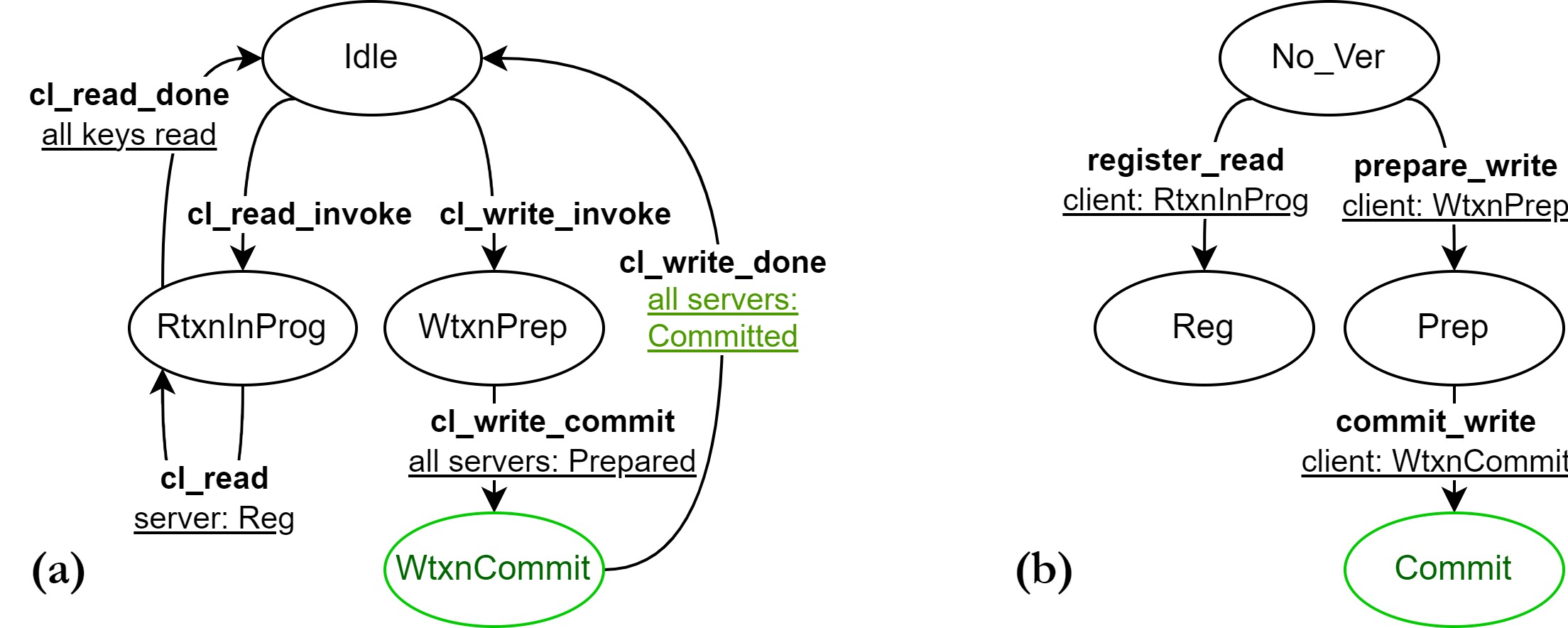}
\end{center}
\captionsetup{skip=0pt}
   \caption{\ouralg: state diagrams of (a) a client's $\inlisa{cl\_state}$ and (b) a server's $\inlisa{svr\_state}$ for a given transaction.  
   }
   \label{fig:ep+-states}
\end{figure}

We formally model \ouralg (\Cref{ssec:formalizing-ep+}),
and use our proof technique~(\Cref{ssec:epp-proof-technique}) to verify its TCCv isolation guarantee (\Cref{ssec:epp-refinement-mapping,ssec:epp-correctness}).

\subsection{Formalizing \ouralg}
\label{ssec:formalizing-ep+}

\looseness=-1
To model the protocol, we consider a distributed KVS with one transaction coordinator per client and several servers that handle the clients' transactions. For simplicity, we integrate the coordinator into the client and assume that each server manages one key. 
We also assume clients execute transactions sequentially.

We formalize \ouralg as an LTS.Its states are the protocol's global configurations, consisting of the clients' and the servers' local configurations. 
As depicted in \Cref{fig:ep+-states},
these local configurations include control states indicating a protocol execution's progress. Each event changes either one client's or one server's configuration and advances its respective control state, ensuring that the clients and servers are independent components with interleaved events. 
We allow these components to directly access each other's local configurations to exchange information. This is a standard abstraction in protocol modeling, which can later be refined into explicit message-passing communication.

We next define our LTS model's configurations and describe the sequences of events associated with read and write transactions in more detail. We slightly deviate from the Isabelle syntax to stay closer to standard mathematical notation.

\begin{figure}[t]
\centering
\begin{minipage}{0.515\textwidth}
\centering
\begin{lstlisting}[numbers=none]
--- transaction state
datatype 'v txn_state = 
 Idle |  
 RtxnInProg ts (key set) (key$\;\map\;$'v) | 
 WtxnPrep (key$\;\map\;$'v) | 
 WtxnCommit ts (key$\;\map\;$'v)

--- version state
datatype 'v ver_state =
 No_Ver |
 Reg |
 Prep ts ts 'v |
 Commit ts ts ts 'v (txid$\;\map\;$ts$\;\times\;$ts)
\end{lstlisting}
\end{minipage}
\vspace{-5pt}
\hspace{4pt}
\begin{minipage}{0.437\textwidth}
\centering
\begin{lstlisting}[numbers=none]
--- client configuration
record 'v cl_conf =
 cl_state : 'v txn_state
 cl_sn : sqn
 cl_clock : ts                       
 gst : ts
 lst_map : key $\fun$ ts
 
--- server configuration
record 'v svr_conf =
 svr_state : txid$\;\fun\;$'v ver_state
 svr_clock : ts
 lst : ts
    
\end{lstlisting}
\end{minipage}
\vspace{-5pt}
\begin{minipage}{0.8\textwidth}
\centering
\begin{lstlisting}[numbers=none]
--- global configuration
record 'v global_conf =
 cls : cl_id $\fun$ 'v cl_conf
 svrs : key $\fun$ 'v svr_conf
 rtxn_rts : txid $\map$ ts        --- three history variables
 wtxn_cts : txid $\map$ ts         
 cts_order : key $\fun$ $\listtype{txid}\;$   
\end{lstlisting}
\end{minipage}
\captionsetup{skip=5pt}
\caption{\ouralg: the client, server, and global configurations.}
\label{fig:ep+-configs}
\end{figure}

\medskip
\inlsec{Configurations}  
We model the client, server, and global configurations as records in Isabelle/HOL (\Cref{fig:ep+-configs}). The global configuration contains the client (\inlisa{cls}) and server (\inlisa{svrs}) local configurations, whose types are parameterized by the type \inlisa{'v} of values, and three history variables, which we discuss in \Cref{sec:hist}.

Besides the global safe time, \inlisa{gst}, and the \inlisa{lst\_map}, already discussed above, the client configuration consists of its state, \inlisa{cl\_state}, a transaction sequence number, \inlisa{cl\_sn}, and the client's (Lamport) clock, \inlisa{cl\_clock}. The state is described by the type \inlisa{txn\_state}, which has four constructors for idle (\inlisa{Idle}), read transaction in progress (\inlisa{RtxnInProg}), write transaction prepare (\inlisa{WtxnPrep}), and write transaction commit (\inlisa{WtxnCommit}). The latter three states include a key-value map describing the values (to be) read or written.

A server's configuration consists of a function mapping each transaction ID to a version state (\inlisa{ver\_state}), the server's (Lamport) clock, and its local safe time (\inlisa{lst}). A version state may either be idle (\inlisa{No\_Ver}), registered read (\inlisa{Reg}) for a read transaction, or prepared (\inlisa{Prep}) or committed (\inlisa{Commit}) for a write transaction. The latter two states include timestamps in their parameters and the commit state includes a readermap to record information about the transactions reading this version (similar to the abstract model's reader sets).

\smallskip
\emph{Notation.} To improve readability, we sometimes omit the projections \inlisa{cls} and \inlisa{svrs}, writing, e.g., \inlisa{(gst$\;$s$\;$cl)} for \inlisa{(gst$\;$(cls$\;$s$\;$cl))}.

\medskip
\inlsec{Read-only transactions} proceed as follows (cf.~\Cref{fig:ep+-states}). 
The \inlisa{cl\_read\_invoke} event of a client starts a read-only transaction and transitions from \inlisa{Idle} to \inlisa{(RtxnInProg clk keys $\emptyset$)} state, where \inlisa{clk} is the client's current clock reading, \inlisa{keys} is the (finite and non-empty) set of keys to be read, and $\emptyset$ is the empty key-value mapping, where the subsequently read values will be recorded. As mentioned, this event also updates the  client's global safe time, \inlisa{gst}, to the minimum of the servers' local safe times stored in \inlisa{lst\_map}. As a result, more up-to-date versions of certain keys may become visible to the client.  

\looseness=-1
Once a client has invoked a read, the involved servers (in \inlisa{keys}) follow with a \inlisa{register\_read} event, where they transition from \inlisa{No\_Ver} to \inlisa{Reg} state and access the client's \inlisa{gst} to determine the latest own write newer than the \inlisa{gst} (if any), or the latest transaction with a commit timestamp \inlisa{cts $\leq$ gst}. This transaction is recorded in the read version's readermap along with the current \inlisa{lst} and updated server clock. 
This information is then accessed in the client's subsequent \inlisa{cl\_read} event, reading the version's value and updating its own clock and \inlisa{lst\_map}.
When the client has read all requested values, i.e., \inlisa{dom kv\_map$\;$=$\;$keys} holds for its state \inlisa{(RtxnProg clk keys kv\_map)}, the event \inlisa{cl\_read\_done} brings it back to  \inlisa{Idle}. 

\medskip
\inlsec{Write-only transactions} are initiated by the \inlisa{cl\_write\_invoke} event, in which the client transitions from \inlisa{Idle} to the state \inlisa{(WtxnPrep kv\_map)}. The key-value map \inlisa{kv\_map} describes the keys and associated values to be written and corresponds to the transaction's (write-only) fingerprint (\Cref{subsec:abstract-transaction-model}). Once all servers have followed into their prepared state, the client can execute the \inlisa{cl\_write\_commit} event, which is defined in \Cref{fig:epp-client-commit}. 

\begin{figure}[t]
\centering
\begin{minipage}{0.94\textwidth}
\begin{lstlisting}
definition cl_write_commit cl kv_map cts sn u clk s s' $\longleftrightarrow$
  --- guards:
  cl_state s cl = WtxnPrep kv_map $\land$
  ($\forall$k $\in$ dom kv_map. is_prepared (svr_state s k (Tn sn cl)) $\land$
  cts = Max {get_prepared_ts (svr_state s k (Tn sn cl)) $\mid$ 
             k $\in$ dom kv_map} $\land$
  sn = cl_sn s cl $\land$
  u = $\refmapviews{\_\epp{}}$ s cl $\land$
  clk = cts + 1 $\land$
  --- updates (unmentioned variables remain unchanged):
  cl_state s' cl = WtxnCommit cts kv_map $\land$
  cl_clock s' cl = clk $\land$
  wtxn_cts s' = (wtxn_cts s)(Tn sn cl $\mapsto$ cts) $\land$
  cts_order s' = extend_cts_order s (Tn sn cl) cts kv_map
\end{lstlisting}
\end{minipage}
\captionsetup{skip=0pt}
\caption{\ouralg's client commit event.}
\label{fig:epp-client-commit}
\end{figure}

\looseness=-1
This event has eight parameters:
the client ID \inlisa{cl}, 
the key-value map \inlisa{kv\_map}, 
the commit timestamp \inlisa{cts}, which is the maximum of the involved server's prepared timestamps for the current transaction \inlisa{(Tn\;sn\;cl)} (lines~5-6), 
the client's current sequence number \inlisa{sn} (line~7),
the abstract view \inlisa{u} (line 8, see \Cref{ssec:epp-refinement-mapping}),
the updated Lamport clock \inlisa{clk} (line~9), and the global states \inlisa{s} and \inlisa{s'} before and after the event. 
As \inlisa{cts\;>\;cl\_clock} always holds here, there is no need to take their maximum to determine \inlisa{clk}.
The guards at lines~3 and~4 require that the client is in the prepared state and that all involved servers have followed into their own prepared state.
The client's state is updated to \inlisa{(WtxnCommit\;cts\;kv\_map)} (line~11) and the clock is updated (line~12). We will discuss the history variable updates at lines 13-14 in \Cref{ssec:epp-refinement-mapping}.  

After the client's commit event, the involved servers commit the transaction on their side using  \inlisa{commit\_write} events. When all servers have committed, the client executes the \inlisa{cl\_write\_done} event to return to its idle state.

\subsection{Proof Technique}
\label{ssec:epp-proof-technique}

We now discuss our protocol verification technique based on refinement, and a technique for commuting independent events required to complement refinement.

The main goal of our verification is to prove the following result, stating that all of \ouralg's reachable states are allowed by the abstract model $\isomodel{\TCCv}$:
\begin{equation}
\label{eq:reach-inclusion}
\refmap{r}{\epp}~(\reach~\epp) \subseteq \reach~\isomodel{\TCCv}.
\end{equation}
Recall from \Cref{subsec:lts-refinement}, that this would follow from a proof of $\epp{} \refines_{\refmap{r}{\epp},\pi_{\,\epp}} \isomodel{\TCCv}$, for suitable refinement mappings $\refmap{r}{\epp}$ and $\pi_{\,\epp}$ on protocol states and events. 
However, such a direct proof would fail for the following reason.

\looseness=-1
\ouralg uses timestamps to define an order on versions and to identify ``safe-to-read'' versions. 
Hence, to ensure that clients always read the latest version in their view, the refinement mapping must reconstruct the version lists of the abstract KVS in the order of their commit timestamps. Otherwise, the proof of the abstract guard $\LWW$ (\Cref{subsec:abstract-transaction-model}) will fail. 
However, the \emph{execution order} of commits and the \emph{order of the associated commit timestamps} may not coincide. We call such commits \emph{\invcmts}. 
Having \invcmts in an execution may require \emph{inserting} a key's new version to its version list rather than \emph{appending} it. Since the abstract model only ever appends new versions at the end of the version lists, the refinement proof alone would fail for executions with \invcmts. 

To address this problem, we introduce an extra proof step to reorder the \invcmts before the refinement. 
To this end, we define a modified protocol model $\widehat{\epp{}}$ 
that restricts transaction commits to those that do not introduce any \invcmts. We decompose the proof of \eqref{eq:reach-inclusion} into the following two steps:
\begin{align}
\label{eq:reduction-guarantee}
\reach~\epp{} & = \reach~\widehat{\epp{}}, \\
\label{eq:reach-inclusion-restricted}
\refmap{r}{\epp}~(\reach~\widehat{\epp{}}) & \subseteq \reach~\isomodel{\TCCv}. 
\end{align}
\looseness=-1
We prove \eqref{eq:reach-inclusion-restricted} by the refinement $\widehat{\epp{}} \refines_{\refmap{r}{\epp},\pi_{\,\epp}} \isomodel{\TCCv}$, which works for the restricted model. To prove \eqref{eq:reduction-guarantee}, we use a proof technique based on Lipton's reduction method~\cite{DBLP:journals/cacm/Lipton75} to successively reorder \invcmts 
in executions by commuting causally independent events, while preserving the executions' final state.

\subsection{Refinement Mapping}
\label{ssec:epp-refinement-mapping}

For the refinement proof, we need to find the refinement mappings $\pi_{\;\epp}$ and $\refmap{r}{\epp}$.
To define $\pi_{\epp}$, we must identify protocol events that refine the abstract commit event and the abstract view extension event. The protocol events \inlisa{cl\_read\_done} and \inlisa{cl\_write\_commit} refine the abstract commit event since reads and writes are guaranteed to commit after these events are executed. The event \inlisa{cl\_read\_invoke} refines the abstract view extension event, as this event updates a client's $\inlisa{gst}$ and thus extends its view. All other events refine $\Skip$. 

To define~$\refmap{r}{\epp}$, we must reconstruct an abstract configuration from \ouralg's protocol configuration $s$, i.e.,
\[
\refmap{r}{\epp}\;s = (\refmapkvs{}\;s,\;\refmapviews{}\;s),
\]
where the function $(\refmapkvs{}$\;$s)$ reconstructs the abstract KVS and $(\refmapviews{}$\;$s)$ reconstructs the abstract client views. To help define these components, we add \emph{history variables} to the global configuration. 
We now describe these history variables and then the functions $\refmapkvs{\_\epp}$ and $\refmapviews{\_\epp}$.

\smallskip
\inlsec{History variables}
\label{sec:hist} 
\looseness=-1
The global configuration includes three history variables. 
The variables \inlisa{rtxn\_rts} and \inlisa{wtxn\_cts} serve as shortcuts to respectively map transaction IDs directly to the read timestamp (\inlisa{gst}) and commit timestamps of the corresponding read-only and write-only transactions. These variables get updated in the corresponding commit events. 
The variable \inlisa{cts\_order} maps each key to a list of client-committed write-only transactions,  \emph{ordered by their commit timestamps}. The client commit event extends \inlisa{cts\_order} by \emph{inserting} the committed transaction's ID at the position corresponding to its commit timestamp into the transaction ID list of each key written by the transaction (line~13 in \Cref{fig:epp-client-commit}). 
This variable is used to facilitate the reconstruction of the abstract KVS.

\smallskip
\inlsec{Abstract KVS} 
The function $\refmapkvs{\_\epp{}}$ reconstructs the abstract KVS from the \inlisa{cts\_order} history variable by mapping each key's list of client-committed transactions to an abstract version list. For a given key \inlisa{k} and transaction \inlisa{t} in the list, we extract each version's value and 
readerset from server \inlisa{k}'s (prepared or committed) state for \inlisa{t}. 
Since the abstract model always reads the latest versions in a client's view, as expressed by the guard \inlisa{\LWW} of the abstract commit event (see \Cref{subsec:abstract-transaction-model}), the \inlisa{cts\_order} must be sorted by commit timestamp.

\smallskip
\inlsec{Abstract views} 
We define the function $\refmapviews{\_\epp{}}$,  reconstructing the abstract views from  \ouralg's model configurations in two steps. We first construct a function \inlisa{get\_view}, where \inlisa{(get\_view s cl k)} denotes the set of client-committed transactions~\inlisa{t}, whose commit timestamp is less than or equal to the client's \inlisa{gst}, i.e., \inlisa{(wtxn\_cts s t) $\leq$ (gst (cls s cl))} or that are the client \inlisa{cl}'s own transactions (for RYW).
Second, we use \inlisa{cts\_order} to map the transactions IDs in the range of \inlisa{get\_view} to their positions in the \inlisa{cts\_order}, which correspond to indices into the abstract version lists.

\subsection{Correctness: \ouralg satisfies \TCCv}
\label{ssec:epp-correctness} 
%
We can now state our main result of this section.
\begin{theorem}[Correctness of \ouralg]
\label{thm:epp-correct}
The Eiger-PORT+ model \epp{} satisfies $\TCCv$, i.e., $\refmap{r}{\epp}~(\reach~\epp) \subseteq \reach~\isomodel{\TCCv}$. 
\end{theorem}

As described in \Cref{ssec:epp-proof-technique}, we combine refinement and reduction in this proof. We devote the remainder of this subsection to sketching both parts of our proof, stated as \Cref{lem:epp-reduction,lem:epp-refinement} below, followed by describing the invariants used in these proofs.

\smallskip
\inlsec{Restricted model and reduction proof} 
We define the restricted model $\eppordered{}$ by adding a guard to the event \inlisa{cl\_write\_commit}, which requires that the unique commit timestamp \inlisa{(cts,$\;$cl)} of the client's transaction is greater than any commit timestamp of a transaction in the \inlisa{cts\_order}. This ensures that the client commit only appends, but does not insert, the new transaction into the \inlisa{cts\_order}.
For this model, we prove the following lemma.

\begin{lemma}
\label{lem:epp-reduction}
$\reach~\inlisa{\epp{}} = \reach~\inlisa{\eppordered{}}$.
\end{lemma}

\begin{proof}[sketch] 
By construction of \eppordered{}, the inclusion ``$\supseteq$'' is easily shown by a refinement. 
For the inclusion ``$\subseteq$'', consider any execution $e$ of \epp{} ending in some state $s$.  We prove by reduction that we can reorder all \invcmts in $e$ of \epp{}, while preserving its final state $s$. We first show that the relevant events of two transactions with inverted commit timestamps are pairwise causally independent. We then prove that adjacent causally independent events in $e$ can be commuted. Using a measure function on executions, we show that this process terminates in an execution $\restricted{e}$ without \invcmts ending in state $s$, thus an execution of \eppordered{}. Hence, any state reachable in \epp{} is also reachable in \eppordered{}. \qed
\end{proof}

\inlsec{Refinement proof} 
Next, we establish a refinement between the restricted model $\eppordered$ and the abstract model $\isomodel{\TCCv}$ instantiated to $\TCCv$, using the refinement mapping defined in \Cref{ssec:epp-refinement-mapping}.

\begin{lemma}
\label{lem:epp-refinement}
$\refmap{r}{\epp}~(\reach~\eppordered) \subseteq \reach~\isomodel{\TCCv}$.
\end{lemma}
\begin{proof}[sketch] 
We show guard strengthening and update correspondence for every event of \eppordered{}. This is easy for most events, which refine $\Skip$. The interesting cases are the read invoke event, which refines the abstract view extension event, and the client commit and read done events, refining the abstract commit event. 

We focus here on the client commit event. The update correspondence proof relies on the absence of \invcmts in \eppordered{} and thus client commits appending versions to KVS version lists.
For guard strengthening, we must show that all guards of the abstract commit event (cf.~\Cref{subsec:abstract-transaction-model}) are implied by the concrete guards. We discuss the most interesting ones. View atomicity (part of view wellformedness) holds by construction, since all versions of a transaction have the same \inlisa{cts} and the abstracted view includes all transactions with a \inlisa{cts} below the client's \inlisa{gst} (and also its own writes).
Similarly, we prove that $\LWW$ holds, i.e., a client reads the latest version in its view.
To show that $\canCommit$ holds, we prove an invariant stating that the clients' views remain closed under $\SO \cup \WR_{\kvs}$. This proof in turn requires several invariants about timestamps. \qed
\end{proof}

\inlsec{Invariants and lemmas} \label{subsubsec:invariants-ep+}
Our proofs rely on numerous invariants and lemmas. We present the most important ones categorized as follows.
\begin{itemize}
    \item \textbf{Freshness of transaction IDs}: The clients' current transaction ID is fresh, i.e., does not occur in the KVS until the commit.
    \item \textbf{Past and future transactions}: stating that the respective client and servers are in particular start states (e.g., \inlisa{Idle}) or end states (e.g., \inlisa{Commit}).
    \item \textbf{Views}: These invariants include view wellformedness, view closedness (for $\canCommit_{\TCCv})$, and session guarantees (monotonic reads and $\RYW$).
    \item \textbf{Timestamps}: This category includes lemmas showing the monotonic increase of timestamps and the following invariant for any client $\inlisa{cl}$ and server~$\inlisa{k}$: 
    \begin{center}
        \inlisa{gst s cl} < \inlisa{lst\_map s cl k} < \inlisa{lst s k} < \inlisa{svr\_clock s k}.
    \end{center}
    This invariant states that the following timestamps are in a strictly increasing order: client $\inlisa{cl}$'s global safe time, client $\inlisa{cl}$'s entry for server $\inlisa{k}$ in its map of local safe times, server $\inlisa{k}$'s local safe time, and $\inlisa{k}$'s local clock value.

    \item \textbf{Client commit order}: 
    The \inlisa{cts\_order} history variable is sorted by commit timestamps and contains only client-committed (and distinct) transactions.   
\end{itemize}

Note that the first three categories are generic and many of their invariants directly imply related guards needed in the abstract commit event's refinement. The last two categories can easily be adapted to other timestamp-based protocols.

\section{Deployment and Evaluation}
\label{sec:deployment-evaluation}

\looseness=-1
We have implemented and deployed our \ouralg protocol on a cluster for a comprehensive performance evaluation~\cite{ghasemirad_2025_14622074}. 
\ouralg pushes the limit of the state-of-the-art, with \emph{superior performance} and a \emph{stronger isolation guarantee}.


\begin{figure*}
  \centering
    \resizebox{.6\linewidth}{!}{
         \pgfplotsset{%
eiger legend/.style={legend image code/.code={%
\node[##1,anchor=west] at (0cm,0cm){\pgfuseplotmark{o}};
\path[#1](0.25cm,-0.05cm)rectangle(.38cm,.15cm);
}},
ep legend/.style={legend image code/.code={%
\node[##1,anchor=west] at (0cm,0cm){\pgfuseplotmark{square}};
\path[#1](0.25cm,-0.05cm)rectangle(.38cm,.15cm);
}},
epp legend/.style={legend image code/.code={%
\node[##1,anchor=west] at (0cm,0cm){\pgfuseplotmark{triangle}};
\path[#1](0.25cm,-0.05cm)rectangle(.38cm,.15cm);
}},
 }

\begin{tikzpicture}
\centering
  \begin{axis}[
      hide axis,
      width=2cm, 
      height=2cm, 
      legend style={at={(0.5,-0.1)},anchor=south, legend columns=3, draw=none},
      /tikz/every even column/.append style={column sep=0.2cm}
  ]
  \addlegendimage{color=blue,pattern color=blue,pattern=crosshatch dots, eiger legend}
  \addlegendentry{\footnotesize Eiger}
  \addlegendimage{color=orange,pattern color=orange,pattern=north west lines,
    ep legend}
  \addlegendentry{\footnotesize Eiger-PORT}
  \addlegendimage{color=purple, fill=purple, epp legend}
  \addlegendentry{\footnotesize Eiger-PORT+}
  \addplot[draw=none] coordinates {(0,0)}; 
  \end{axis}
\end{tikzpicture}
   }
  \resizebox{\linewidth}{!}{
      \raisebox{-0.18cm}{\begin{tikzpicture}[scale=1]
    \begin{axis}[
        title={(a)},
        tick style={draw=none},
        xlabel={Tput (txns/s)},
        ylabel={Latency (ms)},
        xlabel style={yshift=-5pt},
        ylabel style={yshift=8pt},
        ymax=200,
        cycle multiindex* list={
            color list,
            mark list*,
        },
        width=3in, 
        height=2in, 
        title style={font=\normalfont\huge\bfseries},
        tick label style={font=\normalfont\LARGE},
        label style={font=\normalfont\huge},
    ]
    \addplot[color=blue,mark=o,mark size=3pt]
        table [x=throughput,y=latency, col sep=comma]
         {plots/data/NumClients/EIGER.csv};
    
    \addplot[color=orange,mark=square,mark size=3pt]
        table [x=throughput,y=latency, col sep=comma]
         {plots/data/NumClients/EIGER_PORT.csv};
    
    \addplot[color=purple,mark=triangle,mark size=3pt]
        table [x=throughput,y=latency, col sep=comma]
         {plots/data/NumClients/EIGER_PORT_PLUS.csv};
    
    \end{axis}
    \end{tikzpicture}}
      \qquad \qquad
        \begin{tikzpicture}[scale=1]
    \begin{axis}[
        title={(b) },
        tick style={draw=none},
        xlabel={Number of Clients},
        ylabel={Tput (txns/s)},
        xlabel style={yshift=-5pt},
        ylabel style={yshift=8pt},
        ymax=30000,
        cycle multiindex* list={
            color list,
            mark list*,
        },
        width=3in, 
        height=2in, 
        legend style={at={(.3,.94)}, anchor=north},
        title style={font=\normalfont\huge\bfseries},
        tick label style={font=\normalfont\LARGE},
        label style={font=\normalfont\huge},
    ]
    \addplot[color=blue,mark=o,mark size=3pt]
        table [x=num_clients,y=throughput, col sep=comma]
         {plots/data/NumClients/EIGER.csv};
    
    \addplot[color=orange,mark=square,mark size=3pt]
        table [x=num_clients,y=throughput, col sep=comma]
         {plots/data/NumClients/EIGER_PORT.csv};
    
    \addplot[color=purple,mark=triangle,mark size=3pt]
        table [x=num_clients,y=throughput, col sep=comma]
         {plots/data/NumClients/EIGER_PORT_PLUS.csv};
    
    \end{axis}
    \end{tikzpicture}  \qquad \qquad
         \begin{tikzpicture}[scale=1]
    \begin{axis}[
        title={(c)},
        tick style={draw=none},
        xlabel={Number of Servers},
        ylabel={Tput (txns/s)},
        ymax=30000,
        cycle multiindex* list={
            color list,
            mark list*,
        },
        xlabel style={yshift=-5pt},
        ylabel style={yshift=8pt},
        width=3in, 
        height=2in, 
        legend style={at={(.3,.94)}, anchor=north},
        title style={font=\normalfont\huge\bfseries},
        tick label style={font=\normalfont\LARGE},
        label style={font=\normalfont\huge},
    ]
    \addplot[color=blue,mark=o,mark size=3pt]
        table [x=num_servers,y=throughput, col sep=comma]
         {plots/data/NumServers/EIGER.csv};
    
    \addplot[color=orange,mark=square,mark size=3pt]
        table [x=num_servers,y=throughput, col sep=comma]
         {plots/data/NumServers/EIGER_PORT.csv};
    
    \addplot[color=purple,mark=triangle,mark size=3pt]
        table [x=num_servers,y=throughput, col sep=comma]
         {plots/data/NumServers/EIGER_PORT_PLUS.csv};
    
    \end{axis}
    \end{tikzpicture} 
   }
   \resizebox{\linewidth}{!}{
   }
   \resizebox{\linewidth}{!}{    
    \begin{tikzpicture}[scale=1]
    \begin{axis}[
        title={(d)},
        tick style={draw=none},
        xlabel={Skew Factor},
        ylabel={Normalized Tput},
        xlabel style={yshift=-5pt},
        ylabel style={yshift=8pt},
        ymax=2.58,
        ybar,
        bar width=0.17cm,
        cycle multiindex* list={
            color list,
            mark list*,
        },
        width=3in, 
        height=2in, 
        xtick=data, 
        xticklabels={0,.3,.7,.8,.9,.99, 1.1},
        legend style={at={(.28,.93)}, anchor=north, fill=none},
        title style={font=\normalfont\huge\bfseries},
        tick label style={font=\normalfont\LARGE},
        label style={font=\normalfont\huge},
    ]
    \addplot[color=blue,pattern color=blue,pattern=crosshatch dots]
        table [x expr=\coordindex,y=normalized_throughput, col sep=comma]
         {plots/data/ZipfNormalized/EIGER.csv};
    \addplot[color=orange,pattern color=orange,pattern=north west lines]
        table [x expr=\coordindex,y=normalized_throughput, col sep=comma]
         {plots/data/ZipfNormalized/EIGER_PORT.csv};
    \addplot[color=purple,fill]
        table [x expr=\coordindex,y=normalized_throughput, col sep=comma]
         {plots/data/ZipfNormalized/EIGER_PORT_PLUS.csv};
    \end{axis}
    \end{tikzpicture}\qquad \qquad
        \begin{tikzpicture}[scale=1]
    \begin{axis}[
        title={(e)},
        tick style={draw=none},
        xlabel={Number of Clients},
        ylabel={Latency (ms)},
        xlabel style={yshift=-5pt},
        ylabel style={yshift=8pt},
        ymax=200,
        cycle multiindex* list={
            color list,
            mark list*,
        },
        legend style={at={(.7,.34)}, anchor=north},
        width=3in, 
        height=2in, 
       title style={font=\normalfont\huge\bfseries},
        tick label style={font=\normalfont\LARGE},
        label style={font=\normalfont\huge},
    ]
    \addplot[color=blue,mark=o,mark size=3pt]
        table [x=num_clients,y=latency, col sep=comma]
         {plots/data/NumClients/EIGER.csv};
    
    \addplot[color=orange,mark=square,mark size=3pt]
        table [x=num_clients,y=latency, col sep=comma]
         {plots/data/NumClients/EIGER_PORT.csv};
    
    \addplot[color=purple,mark=triangle,mark size=3pt]
        table [x=num_clients,y=latency, col sep=comma]
         {plots/data/NumClients/EIGER_PORT_PLUS.csv};
    
    \end{axis}
    \end{tikzpicture} \qquad \qquad 
    \begin{tikzpicture}[scale=1]
    \begin{axis}[
        title={(f)},
        tick style={draw=none},
        xlabel={Number of Servers},
        ylabel={Latency (ms)},
        xlabel style={yshift=-5pt},
        ylabel style={yshift=8pt},
        ymax=20,
        cycle multiindex* list={
            color list,
            mark list*,
        },
        width=3in, 
        height=2in, 
        legend style={at={(.7,.35)}, anchor=north},
        title style={font=\normalfont\huge\bfseries},
        tick label style={font=\normalfont\LARGE},
        label style={font=\normalfont\huge},
    ]
    \addplot[color=blue,mark=o,mark size=3pt]
        table [x=num_servers,y=latency, col sep=comma]
         {plots/data/NumServers/EIGER.csv};
    
    \addplot[color=orange,mark=square,mark size=3pt]
        table [x=num_servers,y=latency, col sep=comma]
         {plots/data/NumServers/EIGER_PORT.csv};
    
    \addplot[color=purple,mark=triangle,mark size=3pt]
        table [x=num_servers,y=latency, col sep=comma]
         {plots/data/NumServers/EIGER_PORT_PLUS.csv};
    
    \end{axis}
    \end{tikzpicture}   
   }

\captionsetup{skip=5pt}
  \caption{
  Performance comparison among the Eiger-family protocols.
  }
  \label{fig:EP+ results}
  \vspace{-2ex}
\end{figure*}

\looseness=-1
\inlsec{Deployment}
We implement  \ouralg, along with Eiger and \ep, 
in the same codebase, 
each consisting of around 12 kLoC in Java.
We use \ep's 
 workload generator with default parameters of 32 threads per client, 1 million keys, 
 90\% read proportion, and the Zipfian key-access distribution with a skew factor of 0.8.
We deploy these three protocols on a CloudLab~\cite{Duplyakin+:ATC19} cluster of machines, each with 2.4 GHz Quad-Core Xeon CPU and 12 GB
RAM.
By default, we use eight 
servers to partition the database and
eight client machines to load the servers. 
We plot each data point using the average over five 60-second trials.

\smallskip
\inlsec{Evaluation}
Overall, \ouralg is highly performant and superior to both competing protocols (Figure~\ref{fig:EP+ results}a).
In particular, compared to the performance-optimal transaction protocol  \ep  providing TCC,  \ouralg exhibits  \emph{higher} throughput with a \emph{stronger} isolation guarantee including convergence.
\ouralg also scales well with an increasing number of clients (Figure~\ref{fig:EP+ results}b) and servers (Figure~\ref{fig:EP+ results}c), with 
up to 1.8x (resp. 2.5x) throughput improvement  
over \ep (resp. Eiger).
In addition, despite varying skews, \ouralg's throughput consistently surpasses that of its competitors  (Figure~\ref{fig:EP+ results}d).
This improvement becomes more pronounced with highly skewed workloads (or larger skew factors). 
This is because higher skewness results in increased concurrency, which would trigger additional server-side computation in \ep
and more rounds of communication in Eiger. 
Notably, despite its higher throughput and stronger isolation guarantee, \ouralg demonstrates similar latency to \ep (Figures~\ref{fig:EP+ results}e and~\ref{fig:EP+ results}f), which has been proven to be latency-optimal~\cite{SNOW:OSDI2016}.

\section{Related Work}

We compare our work with other efforts on verifying transaction protocols, distinguishing them based on testing, model checking, and deductive verification.

\smallskip
\inlsec{Testing} 
Previous work has devised various testers either for specific isolation levels, such as SI~\cite{polysi} and SER~\cite{cobra}, or for a range of levels~\cite{isovista,txcheck,elle,plume}. The underlying techniques are usually based on a characterization of anomalies, e.g., specified using dependency graphs~\cite{adya1999weak} or axioms~\cite{DBLP:journals/pacmpl/BiswasE19,DBLP:journals/jacm/CeroneG18}. 
In contrast to our work, testing can only verify individual protocol executions and therefore can easily miss rarely occurring isolation bugs. On the other hand, while we are verifying a protocol model, testing can be done on the actual implementation.

\smallskip
\inlsec{Model checking}
The Maude model checker has been used to verify the RAMP and LORA protocols for RA~\cite{DBLP:conf/sac/LiuORGGM16,10.1145/3494517}, the Walter protocol for (parallel) SI~\cite{DBLP:conf/wrla/LiuOWM18}, the ROLA protocol for Update Atomicity~\cite{DBLP:journals/fac/LiuOWGM19}, and the MegaStore and P-Store protocols for SER~\cite{DBLP:conf/birthday/GrovO14,DBLP:conf/wadt/Olveczky16}. Maude has also been used to verify various non-transactional consistency properties of the Cassandra key-value store~\cite{DBLP:conf/icfem/LiuRSGM14}.
Using the TLA+ model checker, the Azure CosmosDB has been verified against several non-transactional consistency properties~\cite{Azure} and TiDB against SI~\cite{TiDB}. This model checker has also been used to verify some properties of concurrency control protocols other than isolation guarantees~\cite{DBLP:conf/eurosys/KatsarakisMTBBD21,tapir}.
For model checking to be feasible, one must usually impose certain bounds (e.g., on the number of processes and transactions). In contrast, our work provides fully general verification results, which hold for arbitrary protocol executions.

\smallskip
\inlsec{Deductive verification}
Xiong et al.~\cite{DBLP:conf/ecoop/XiongCRG19} also combine reduction and refinement of their abstract transaction model to prove that the COPS protocol~\cite{COPS:SOSP2011} (with read-only transactions but single writes) 
satisfies TCCv and the Clock-SI protocol~\cite{DBLP:conf/srds/DuEZ13} satisfies SI. However, while being general, pen-and-paper proofs of such complex protocols are error-prone.
Previous work on mechanized deductive verification, to the best of our knowledge, only covers either non-transactional consistency properties or serializability of textbook protocols.
Chapar is a framework for verifying causal consistency of non-transactional KVSs.
PVS and Event-B have been applied to verify (S)SER of the two-phase locking (2PL) protocol~\cite{DBLP:conf/icfem/ChkliaevHS00,DBLP:conf/rodin/YadavB06}.
In contrast, we have designed a new protocol, \ouralg, which is substantially more complex than 2PL, and we have verified that it satisfies TCCv. 

\section{Conclusion}

\looseness=-1
We have designed \ouralg, a novel, causally-consistent, database transaction protocol, and formally verified its isolation guarantee of TCCv in Isabelle/HOL. In particular, TCCv was previously conjectured to be incompatible with transactional writes in the presence of performance-optimal read-only transactions.
We have formally refuted this conjecture by our protocol design and its verification. Moreover, this case study represents the first complete formal verification of a complex distributed database transaction protocol.
Our verification effort, excluding the verification framework, amounts to 10.3k lines of Isabelle/HOL code, composed of 0.7 kLoC for the model and 9.6 kLoC for the proof, which required 108 invariants. 
In addition, we have conducted a comprehensive evaluation, demonstrating  \ouralg's superior performance over two state-of-the-art protocols.
We believe our protocol is an attractive choice for database applications opting for TCCv. 

We see several avenues for future work.
First, to facilitate formal protocol modeling and correctness proofs, we will develop an abstract distributed protocol model as an intermediate refinement step. This model will capture structure common to protocols and factor out recurring parts of correctness proofs.
Second, we intend to support additional protocol features, for example, open-loop clients, which optimize transactional writes by immediately starting a new transaction, once they commit the previous one. An interesting case study in this context would be Eiger-NOC2~\cite{noc-noc}, a recent successor of \ouralg that improves its performance by using open-loop clients and other features. However, this protocol has not yet been formally verified to provide $\TCCv$. Third, we envision verifying our implementation and connecting it to our protocol verification results, possibly following the Igloo methodology~\cite{SprengerKEWMCB2020}.

\subsubsection*{Acknowledgements.}
We thank the anonymous reviewers for their valuable feedback. This research is supported by an ETH Zurich Career Seed Award and the Swiss National Science Foundation project 200021-231862 ``Formal Verification of Isolation Guarantees in Database Systems''.

%
%
%
\bibliographystyle{splncs04}
\bibliography{EPplus}

\iffull
\appendix
\section{The Pseudocode of \ouralg}
\label{app:epp}

\subsection{Logical Time}
Logical time refers to a way of keeping track of events in a distributed system where there is no global clock available. It is based on the concept of logical clocks which are used to assign a unique timestamp to each event, ensuring that the order of events is preserved in a consistent and meaningful way. We employ Lamport clocks~\cite{lamportclk}, where each process maintains a local clock 
whose value is increased whenever an event occurs. 
When a process sends a message, it includes its current clock reading as a timestamp, which is then used by the receiving process to update its clock to a value greater than both the received timestamp and its local clock. 
Lamport clocks reflect the causal ordering of events in the sense that the order of the timestamps associated to causally dependent events is consistent with their causal ordering. 
Overall, logical time and Lamport clocks are essential tools for managing distributed systems and ensuring that they operate correctly and efficiently.

\looseness=-1
In our pseudocode, given in the remainder of this section, clients and servers communicate using remote procedure calls (RPCs). These calls' invocation and reply messages use both sides' Lamport clocks as described above. However, to simplify the presentation, we do not explicitly show the use of these clocks on the client side. On the server side, we do show clock readings and updates, using an object \inlisa{clk\;:\;clock} with two methods: \inlisa{current()}, which reads the current clock value, and \inlisa{advance()}, which updates the clock, taking into account its current value and the client's clock value received with the current RPC invocation message.

\subsection{Client-side Procedures}

The pseudocode for the client-side data structures and its read and write procedures are given in \Cref{fig:epp-client}.

\begin{figure}[t]
    \centering
\begin{lstlisting}[backgroundcolor = \color{lightblue}, numbers=left]
--- version record
version = {
  cl_id : client_id,
  val : value,
  is_pend : bool,
  pend_t : timestamp,
  prep_t : timestamp,
  comt_t : timestamp
}

--- data structures:
gst : timestamp              --- global safe time
lst_map : key $\fun$ timestamp   --- maps key to its server's local safe time

--- procedures:
function read_only_txn(keys : key set) : key $\map$ value
  gst = min(range(lst_map))
  for k in keys       --- in parallel
    (vals[k], lst_map[k]) = read(k, gst, cl_id)
  return vals

function write_only_txn(vals : key $\map$ value) : unit
  --- Prepare
  for (k, v) in vals  --- in parallel
    ver[k] = prepare_write(k, v, cl_id)
  
  --- Commit
  commit_t = max {ver[k].prep_t | k in dom(vals)}
  for (k, v) in vals   --- in parallel
    lst_map[k] = commit_write(ver[k], commit_t)
  return ()
\end{lstlisting}
    \caption{Client-side procedures in \ouralg. 
    }
    \label{fig:epp-client}
\end{figure}

\inlsec{Data structures} We represent \ouralg's versions as records of type \inlisa{version}, which have six fields (lines~2-9): the version writer's client ID, a value, a flag \inlisa{is\_pending}, indicating whether the version is pending, and three timestamps, corresponding to its pending, preparing, and commit times. This data structure is also used on the server side. The client keeps track of the global safe time via a timestamp called \texttt{gst} (line~12). The versions with a commit timestamp smaller or equal to the global safe time 
are safe to read from every partition. Furthermore, the client also stores the latest safe time for every partition in the map \texttt{lst\_map}~(line~13). 

\inlsec{Read-only transactions} We first compute the current global safe time as the minimum of the local safe times (line~17). This ensures that it is safe to read at the global safe time from every server. Then, we call the server-side \texttt{read} procedure in parallel for all involved keys, sending to the partition in charge the key, the \texttt{gst}, and a client identifier (lines~18-19). This procedure returns the result value and the new local safe time, which we store in the maps \inlisa{vals} and \inlisa{lst\_map}. Finally, return the key-value map \inlisa{vals} (line~20).

\inlsec{Write transactions} First, we call \texttt{prepare\_write} in parallel for every key-value pair, sending to the respective partition the key, the value, and the (unique) client identifier (lines 24-25). This procedure returns a prepared version, which we store in the map \inlisa{ver}.
We then compute the commit timestamp \texttt{commit\_t} by selecting the largest prepared timestamp (line~28). Finally, we call the server-side procedure \texttt{commit\_write}  in parallel for every key and value pair, providing the respective prepared version and commit timestamp (lines~29-30). We receive as a response a new local safe time, which is used to update \texttt{lst\_map}.

\subsection{Server-side Procedures}

\Cref{fig:epp-server} displays the pseudocode for the server-side data structures and the read and write procedures.

\begin{figure}[t!]
\centering
\begin{lstlisting}[backgroundcolor = \color{lightblue}, numbers=left]
--- data structures:
clk : clock                   --- Lamport clock
lst : timestamp               --- local safe time, updated upon writes
pending_wtxns : timstamp set  --- pending write txns' timestamps
kvs : key $\fun$ version set      --- multi-versioned key-value store

--- read-only procedure:
function read(k : key, rts : timestamp, cl_id : clientID)
             : value $\times$ timestamp
  clk.advance()
  --- largest vers below rst (ordered by commit_t)
  ver = at(kvs[k], rts)
  --- ensure Read-Your-Writes
  for v in sort_decreasing_by_commit_t(kvs[k]) 
    if v.cl_id == cl_id & v.comt_t > rts then
      return (v.val, lst)
  return (ver.val, lst)

--- write-only procedures:
function prepare_write(k : key, v : value, cl_id : client_id) : version 
  pending_t = clk.current()
  clk.advance()
  pending_wtxns.insert(pending_t)
  ver = {val = v, cl = cl_id, is_pend = true, 
         pend_t = pending_t, prep_t = clk.current()}
  kvs[k].insert(ver)
  return ver

function commit_write(ver : version, commit_t : timestamp) : timestamp
  clk.advance()
  kvs[k].remove(ver)
  ver.comt_t = commit_t
  ver.is_pending = false
  pending_wtxns.remove(ver.pend_t)
  kvs[k].insert(ver)
  if pending_wtxns == $\emptyset$
  then lst = clk.current()
  else lst = min(pending_wtxns)
  return lst
\end{lstlisting}
    \caption{Server-side read-only and write-only transactions in \ouralg.}
    \label{fig:epp-server}

\end{figure}

\inlsec{Data structures} In \ouralg, the servers keep track of the local safe time \texttt{lst}, which, if there are any pending transactions in the partition, is the smallest of the pending timestamps, and otherwise is the current value of the Lamport clock. Moreover, we keep a set of timestamps, called \texttt{pending\_wtxns}, containing all pending timestamps of uncommitted writes. Finally, the server maintains a multi-versioned key-value store \texttt{kvs}.

\inlsec{Read} After advancing the Lamport clock (line~10), we first fetch the version \texttt{ver} at the read timestamp \texttt{rts} (line~11), which is the global safe time according to the client that issued the read. This returns the newest version with a commit timestamp smaller than or equal to \texttt{rts}. Then, we look for the newest committed version written by the same client with a commit timestamp greater than \inlisa{rts}, and return it, if it exists (lines 14-16). Otherwise, we return \texttt{ver} (line 17). In either case, we also return the local safe time.

\inlsec{Prepare write} The server uses the current value of the Lamport clock as a timestamp \texttt{pending\_t} for the pending transaction, it adds said value to \texttt{pending\_wtxns} and advances the clock (lines 21-23). Then it adds a new version of the key to the database (lines 24-26). This version contains the value, the client id, the \texttt{is\_pending} flag set to true, the pending timestamp set to \texttt{pending\_t}, and the prepared timestamp set to the new clock value. Finally, the procedure returns the new version.

\inlsec{Commit write} We first advance the Lamport clock and remove the received (pending) version \inlisa{ver} from the \inlisa{kvs} (lines 30-31). Then we update the version by setting its commit timestamp to \texttt{commit\_t} and its \texttt{is\_pending} flag to false (lines 32-33). Next, we remove the version's pending timestamp from the set of pending transactions and add the updated version to the \inlisa{kvs} (lines 34-35). Finally, we update \texttt{lst} to the smallest pending timestamp, if there are any, or otherwise to the updated Lamport clock value. The procedure then returns the new local safe time (lines 36-39).
\fi

\end{document}